\documentclass[reprint,
latex apssamp.tex,
bibtex apssamp
bibnotes,
 amsmath,amssymb,
 aps,
pra, superscriptaddress]{revtex4-2}

\usepackage{tikz}
\usetikzlibrary{arrows}
\usepackage{amssymb,amsmath,amsthm}
\usepackage{xcolor}
\usepackage[colorlinks=true,pdfborder={0 0 0},linkcolor=blue]{hyperref}
\usepackage{graphicx}
\usepackage{dcolumn}
\usepackage{bm}
\usepackage{physics}
\usepackage{braket}
\usepackage[english]{babel}
\usepackage{svg}
\usepackage{subcaption}
\usepackage{multirow}
\newtheorem{theorem}{Theorem}
\newtheorem{definition}{Definition}

\newtheorem{lemma}{Lemma}
\newtheorem{sublemma}{Sublemma}

\newtheorem{fact}{Fact}

\usepackage{url}
\DeclareFixedFont{\ttb}{T1}{txtt}{bx}{n}{10}
\DeclareFixedFont{\ttm}{T1}{txtt}{m}{n}{10}
\usepackage{color}
\usepackage{dsfont}
\usepackage{float}
\usepackage{dsfont}
\usepackage{mathtools}
\usepackage{dcolumn}
\newcolumntype{2}{D{.}{}{2.0}}
\definecolor{deepblue}{rgb}{0,0,0.5} 
\hypersetup{breaklinks=true}

\usepackage{listings}
\begin{document}
\title{Majority-Agreed Key Distribution using Absolutely Maximally Entangled Stabilizer States}

\author{Sowrabh Sudevan}
\email{ss18ip003@iiserkol.ac.in}
\affiliation{Indian Institute of Science Education and Research, Kolkata}

\author{Ramij Rahaman}
\email{ramijrahaman@isical.ac.in}
\affiliation{Physics and Applied Mathematics Unit, Indian Statistical Institute, 203 B.T. Road, Kolkata 700108, India}

\author{Sourin Das}
\email{sourin@iiserkol.ac.in}
\affiliation{Indian Institute of Science Education and Research, Kolkata}

\begin{abstract}

In [Phys. Rev. A 77, 060304(R),(2008)], Facchi et al. introduced absolutely maximally entangled (AME) states and also suggested ``majority-agreed key distribution"(MAKD) as a possible application for such states. In MAKD, the qubits of an AME state are distributed one each to many spatially separated parties. AME property makes it \textit{necessary} that quantum key distribution(QKD) between any two parties can only be performed with the cooperation of a majority of parties. Our contributions to MAKD are, $(1)$ We recognize that stabilizer structure of the shared state is a useful addition to MAKD and prove that the cooperation of any majority of parties(including the two communicants) is \textit{necessary} and \textit{sufficient}. Considering the rarity of \textit{qubit} AME states, we extended this result to the \textit{qudit} case. $(2)$ We generalize to shared graph states that are not necessarily AME. We show that the stabilizer structure of graph states allows for QKD between any inseparable bipartition of qubits. Inseparability in graph states is visually apparent in the connectivity of its underlying mathematical graph. We exploit this connectivity to demonstrate conference keys and multiple independent keys per shared state. 
Recent experimental and theoretical progress in graph state preparation and self-testing make these protocols feasible in the near future.
\end{abstract}

\maketitle

\section{Introduction}
Methods of private communication (cryptography)\cite{Loepp_Wootters_2006,Arora_Barak_2009,Nielsen_Chuang_2010} have been explored since antiquity. These methods and the study of their security were formalized with the discretization of messages into strings of numbers\cite{6769090}. This allowed the task of encryption to be understood in terms of a function or algorithm that takes in two number strings(the encoded plain text and key) and returns the encrypted message. This allows the encryption algorithm to be publicly known, while resting the privacy of the message solely on the key. Secure communication is now a problem of secure key distribution among intended parties. While classical key distribution relies on exploiting the computational hardness of inverting certain trap-door functions\cite{1055638,2021LNP...988.....W}, quantum key distribution\cite{SCARANI201427,Pirandola:20} utilizes physical phenomena\cite{RevModPhys.94.025008,RevModPhys.74.145,Masanes_2009,Colbeck_2011} such as quantum measurements, \cite{portmann2014cryptographicsecurityquantumkey,PhysRevLett.68.3121,PhysRevLett.102.020504}, the non-commutativity of observables and the monogamy of entanglement\cite{5388928,PhysRevLett.113.140501,PhysRevLett.92.217903,cerveromartín2023deviceindependentsecurityquantum},to detect potential eavesdroppers in the public channel\cite{PhysRevA.59.4238,PhysRevA.72.032301,Tomamichel2017largelyself}. 

All known QKD protocols\cite{sabani2022quantum} can be placed in two groups, the prepare and measure ones\cite{10.1145/382780.382781,PhysRevLett.85.441,BENNETT20147,Hayashi_2012} and the entanglement-based ones\cite{2021LNP...988.....W}. In this letter, we focus on the latter. The first cryptographic protocol that explicitly uses entanglement as a resource is the eponymous Ekert-91 (E-91) protocol\cite{PhysRevLett.67.661}. This protocol requires the distribution of the two qubits of many copies of a Bell pair between two communicants. Some fraction of this shared resource is used to detect potential eavesdroppers\cite{vsupic2020self}, and once eavesdroppers are ruled out, the remaining fraction is used to generate a perfectly random bit string known only to the two intended communicants. This bit string serves as the key to be used with OTP\cite{PhysRevLett.98.140502}, AES or other encryption schemes. Since the E-$91$ protocol uses bipartite entanglement, it is natural to explore what practical features can be added to key distribution if one uses multipartite entanglement\cite{RevModPhys.81.865,PhysRevLett.110.030501}. 

Multipartite entanglement, even among pure states, comes in many forms, and the various distinct and sometimes conflicting measures used to quantify this, sometimes capture independent physical properties of the state\cite{bengtsson2016briefintroductionmultipartiteentanglement,PhysRevA.69.052330,PhysRevA.69.062311,PhysRevA.106.032604,PhysRevA.104.032601,PhysRevA.107.022428,Borras_2007,PhysRevA.100.062329,Eltschka_2014,PhysRevA.86.052335,10143323,PhysRevA.108.022426,PhysRevA.69.052330}. Absolutely maximally entangled states\cite{helwig2013absolutelymaximallyentangledstates,Helwig_2012,Enríquez_2016} are a straightforward generalization of maximally bipartite entangled states. In such states, the entanglement entropy across every bipartition of the system is maximum. A consequence of this maximum entanglement is that,
\begin{eqnarray}
    \bra{\psi}\mathcal{\hat{O}}(r) \ket{\psi} = 0,\;\;\forall r\in [1,\lfloor n/2 \rfloor],
\end{eqnarray}
where, $\ket{\psi}$ is an $n$-qubit AME state and $\mathcal{\hat{O}}(r)$ is an observable with support on $r$ qubits.

AME states have been used in the study of quantum error correction, and more pertinently, for us, in Facchi et al.'s majority-agreed key distribution (MAKD). 

\section{Majority-agreed key distribution}
MAKD requires each qubit of an AME state of $n$-qubits be distributed among $n$ parties. Now, QKD between any two parties can be accomplished using the cooperation of a majority of the parties. They illustrated MAKD with the following five qubit AME state,
\begin{equation} \label{Facchi_state}
   \ket{\psi_5} = \frac{1}{\sqrt{2^5}}\sum_{k\in \mathbb{Z}_2^5} e^{i\phi_k} \ket{k}, 
\end{equation}
where $(\phi_k) = (0,0,0,0,0,\pi,\pi,0,0,,\pi,\pi,0,0,0,0,0,0,0, \\\pi,\pi,0,\pi,0,\pi,\pi,0,\pi,0,\pi,\pi,0,0)$. The AME property of $\ket{\psi_5}$ ensures that,
\begin{equation} \label{2-uniform}
    \bra{\psi_5} \hat{O}(r) \ket{\psi_5} = 0 \;\;, \forall r \in \{1,2\}. 
\end{equation}
Equation \ref{2-uniform} implies that the measurement outcomes of any one party is identical to a fair coin toss and also that any choice of measurements on any two parties is completely uncorrelated. Therefore one should consider correlations between observables on more than two parties. In fact, they note the following perfect correlation between Pauli measurements on all the parties,
\begin{equation}\label{Facchi_correlation}
    \bra{\psi_5} Z_1 Z_2 Y_3 Y_4 Z_5 \ket{\psi_5} = 1,
\end{equation}
and claim that key distribution between any two parties requires the blind cooperation of the other three. For example, if party-$1$ announces the intent to set up a secure channel with party-$2$, the other three parties can choose to make the measurements $Y_3,Y_4$ and $Z_5$ respectively and announce their outcomes. This reduces correlation \ref{Facchi_correlation} to the binary equation,
\begin{equation} \label{facchi_bin}
     Z_1^{\mathcal{M}} \cdot Z_2^{\mathcal{M}} = Y_3^{\mathcal{M}}\cdot Y_4^{\mathcal{M}}\cdot Z_5^{\mathcal{M}}, 
 \end{equation}
 where, $\sigma_i^{\mathcal{M}}$ represents the result of measuring the Pauli ($\sigma$) by the $i$'th party. Note that the terms in the right hand side of equation \ref{facchi_bin} is public knowledge. Using this value, parties $1$ and $2$ can know each other's measurement outcomes and thereby complete the key exchange.
 
Although QKD between two parties can be accomplished using this correlation, we note that \cite{PhysRevA.77.060304} needs a clarification: the cooperation of all three other parties is not necessary. The state in equation \ref{Facchi_state} has a perfect correlation/anti-correlation involving Pauli measurements for any choice of three qubits as seen in table \ref{Facchi_table}.
 \begin{table}[]
    \centering
    \begin{tabular}{ |c|c|c|c|}
 \hline \hline
  $(i,j,k)$ & Stabilizers & $(i,j,k)$ & Stabilizers \\
 \hline
 $(1,2,3)$ & $Z_1 X_2 X_3$ & $(1,4,5)$ & $Z_1 X_4 X_5$ \\
 $(1,2,4)$ &  $X_1 Z_2 Z_4$ & $(2,3,4)$ & $-Y_2 Y_3 Y_4$ \\
 $(1,2,5)$ & $Y_1 Y_2 Z_5$ & $(2,3,5)$ & $Z_2 Z_3 X_5$ \\
 $(1,3,4)$ & $Y_1 Z_3 Y_4$ & $(2,4,5)$ & $-X_2 Y_4 Y_5$ \\
 $(1,3,5)$ &  $X_1 Y_3 Y_5$ & $(3,4,5)$ & $X_3 Z_4 Z_5$ \\
 
\hline \hline
    \end{tabular}
    \caption{Perfect correlations/stabilizer elements of $\ket{\psi_5}$ for all qubit triples (i,j,k).}
    \label{Facchi_table}
\end{table}
Therefore, any two-party QKD can be accomplished using the cooperation of any one other party. 

The perfect correlations in Table \ref{Facchi_table} are a consequence of the fact that the state in Equation \ref{Facchi_state} happens to be a stabilizer state\cite{gottesman1997stabilizer}. $\ket{\psi}$ is an $n$ qubit stabilizer state if there exist $n$ independent tensor products of Pauli matrices $\{S_i\}$, such that 
\begin{eqnarray*}
 S_i\ket{\psi} = \ket{\psi}\;\; : \;\; S_i S_j = S_j S_i \And S_i^2 = I \\ 
  \implies \bra{\psi} \prod_{i=1}^n S_i^{r_i} \ket{\psi} = +1, \;\; : \;\; r_i \in \{0,1\}.
\end{eqnarray*}
$\prod_{i=1}^n S_i^{r_i}$ are elements of the stabilizer group corresponding to $\ket{\psi}$, $\mathcal{S}(\ket{\psi})$. Therefore, $\ket{\psi}$ has $2^n$ independent perfect correlations involving Pauli measurements on qubits. Any set of measurements that do not correspond to a stabilizer element is completely uncorrelated.

AME structure makes it necessary that more than $\lfloor n/2\rfloor$ parties cooperate since the observations of a smaller number of parties is completely independent of each other. Imposing stabilizer structure, ensures that all non-zero correlations are perfect and potentially usable for key distribution. It may be possible that one could use shared non-stabilizer AME states. Then, the non-zero correlations are not necessarily perfect, nevertheless, it may still be possible to use some combination of partial correlations to perform key exchange. We will not consider this more general problem in this letter. In summary, while not all AME states are stabilizer states, those that are, are easier to study. Hence, we explicitly impose stabilizer structure on our shared AME states and have arrived at the following theorem.

\begin{theorem} \label{AME theorem}
    In stabilizer states that are also AME, for any choice of $\lfloor n/2 \rfloor + 1$ qubits, one can find a stabilizer element with support only on those qubits. In fact, when $n$ is odd/even, we have exactly one/three such stabilizer elements. 
\end{theorem}
See appendix \ref{proof theorem 1} for proof. Therefore, with access to $n$-qubit AME stabilizer states shared between $n$ parties, QKD between any two parties can be performed using the permission of any set of $\lfloor n/2 \rfloor -1$ other parties. The number of cooperators here is both necessary and sufficient. 

Since MAKD involves many measurements and public announcements of the outcomes, one might worry if these public announcements can render the key calculable by the public or by one of the other parties sharing the state. We show in appendix \ref{substring AME} how neither of these are possible.

However, among multi-qubit states, AME states exist only when $n=2,3,5,6$ \cite{PhysRevA.77.060304,PhysRevLett.118.200502}. Many AME states have been discovered for the more general case of multi-qudit states. A table listing all known AME states can be found in \cite{AME_table,Huber_2018}. The rarity of absolutely maximal entanglement among multi-qubit states motivates us to  consider if a version of Theorem \ref{AME theorem} can be true for multi-qudit states. In fact, this is true.
\begin{theorem}\label{qudit ame theorem}
    $n$-qudit stabilizer AME states possess at least one stabilizer with support exactly on any choice of $\lfloor n/2 \rfloor + 1$ qudits.
\end{theorem}
See proof in appendix \ref{qudit ame proof}. Theorem \ref{qudit ame theorem} provides yet another motivation to discover more stabilizer AME states \cite{Raissi:22,PhysRevA.106.062424,PhysRevResearch.2.033411} because of its potential application in MAKD.

\section{Self-testing}
Note that we have only discussed MAKD after the state has been faithfully shared between the spatially separated parties. Distributing the qubits of some state to many spatially separated locations via quantum channels exposes the qubits to environmental decoherence\cite{PhysRevResearch.5.043260,Kofman_2011} and intentional interactions by eavesdroppers. Let us see how this is addressed in E-$91$\cite{PhysRevLett.67.661} (the first entanglement based QKD protocol). Here the Bell state,
\begin{eqnarray}
 \ket{b} = \frac{1}{\sqrt{2}}(\ket{00} + \ket{11})   
\end{eqnarray}
is distributed between parties $1$ and $2$. Consider the CHSH inequality\cite{PhysRevLett.23.880,PhysRevLett.47.460},
\begin{eqnarray}\label{bell inequality}
    I_b := \langle A_1 B_2 \rangle + \langle A_1^\prime B_2\rangle + \langle A_1 B_2^\prime \rangle - \langle A_1^\prime B_2^\prime \rangle. 
\end{eqnarray}
Where $A,A^\prime$ and $B,B^\prime$ are local dichotomic observables for parties $1$ and $2$ respectively. Local realistic assignment of dichotomic values ($\{\pm 1\}$) for the observables in \ref{bell inequality} shows that $|I_b|\leq 2$. Now let,
\begin{eqnarray}
    A, A^\prime \equiv \frac{X\pm Z}{\sqrt{2}} \;\;\And \;\; B,B^\prime \equiv X,Z.
\end{eqnarray}
We see that the quantum expectation values of $I_b$ with the above choice of local observables is $2\sqrt{2}$. In fact, this violation can be used for self-testing, i.e., to confirm that the qubit distribution was performed faithfully. Any reduction in the degree of violation of this Bell inequality quantifies the extend to which the state has decohered or an eavesdropper has attempted to extracted information. One method to implement E-$91$ is to have party $1$ and party $2$ choose uniformly at random from the following sets of measurements,
\begin{eqnarray}
    \text{Party }1\rightarrow\{A,A^\prime\} \; \And\; \text{Party }2\rightarrow\{A,A^\prime,B,B^\prime\}.
\end{eqnarray}
After the measurements are made, the choice of measurements is announced publicly by $1$ and $2$. The measurements performed by $1$ and $2$ match one of the terms in $I_C$ with $50\%$ probability. Therefore, this fraction of the Bell pairs is used for self-testing. Once the identity of the state is verified, the fraction of Bell pairs where the measurements by $1$ and $2$ matched ($25\%$) can be used for QKD since,
\begin{eqnarray}
    \bra{b} A_1 A_2 \ket{b} = \bra{b} A_1^\prime A_2^\prime\ket{b} = 1.
\end{eqnarray}
Many self-testing inequalities are known for multiqubit states that involve local observables\cite{zhao2022constructing,panwar2023elegant,makuta2021self,kalev2019validating,_upi__2016}. In this letter we borrow the scalable Bell inequalities by \cite{PhysRevLett.124.020402} for graph states. Graph states\cite{hein2006entanglement,Epping_2016,9252479,PhysRevLett.129.090501,cui2015generalized,Adcock:2019vri,Lin2024graphiqquantum,Thomas_2024} are stabilizer states where the stabilizers can be constructed with reference to a simple, undirected mathematical graph. For graph with adjacency matrix $G$, the stabilizers of the corresponding graph state $\ket{G}$ are given by,
\begin{eqnarray}
    S_i = X_i \prod_{j\in \mathcal{N}(i)} Z_j,
\end{eqnarray}
where $\mathcal{N}(i)$ is the set of neighbors of vertex $i$ in graph $G$. The inequalities in \cite{PhysRevLett.124.020402} are useful for us because every stabilizer state is equivalent to some graph state up to local unitaries\cite{dahlberg2020transform}. Therefore for any AME, stabilizer state we can consider an LU equivalent graph state for which the Bell tests in \cite{PhysRevLett.124.020402,yang2022testing} can used up to a local basis change. Going back to $\ket{\psi_5}$ from equation \ref{Facchi_state}, we can show that,  
\begin{eqnarray}\label{facchi to c}
    X_1 F_1 G_2 F_3 H_3 Z_4 H_4 \ket{\psi_5} = \ket{\text{c}}
\end{eqnarray}
\begin{eqnarray*}
  \text{where,}\;\; F = \frac{X+Y}{\sqrt{2}}\;;\; G = \frac{Y+Z}{\sqrt{2}} \;;\; H = \frac{Z+X}{\sqrt{2}},  
\end{eqnarray*}
and $\ket{c}$ is the cluster state with periodic boundary conditions or ring graph state on $5$-qubits, with the set of stabilizer generators,
\begin{eqnarray}
    S_i = Z_{i-1} X_i Z_{i+1}.
\end{eqnarray}
$\ket{c}$ can be self-tested using the Bell operator,
\begin{eqnarray}\label{cs inequality}
    I_c = 2\sqrt{2}\langle X_1 A_2^\prime A_5^\prime\rangle + \sqrt{2} \langle Z_1 A_2 A_3^\prime \rangle +\\ \sqrt{2} \langle Z_1 A_4^\prime A_5 \rangle + \langle A_2^\prime A_3 A_4^\prime \rangle+\langle A_3^\prime A_4 A_5^\prime \rangle. 
\end{eqnarray}
Local-realistic(classical) substitution of dichotomic values for each observable gives $|I_c|\leq 6$, whereas the quantum expectation value is $\bra{c.s} I_c \ket{c.s} = 2+4\sqrt{2}$. 
Using $I_c$ and equation.\ref{facchi to c} we can construct the corresponding Bell inequality for $\ket{\psi_5}$ by a simple basis change $I_{\psi_5} = U I_c U^{\dagger}$, where $U = F_1 X_1 G_2 H_3 H_4 Z_4$. Now, we must randomly partition the shared states to a set for self-testing and a set for MAKD. Due to the larger number of observables in \ref{cs inequality}, each party making a random choice between some set of observables will result in very few of the states being used for self-testing or for MAKD. This prompts us to consider partitioning the shared states in a much more obvious way. Assume that $f\in[0,1]$ is the fraction of the shared states that we wish to use for self-testing. Consider an approximation of $f$, i.e., $f\approx k/d$. Now let each party toss a fair $d$-sided die and announce their outcomes $t_i$. With all of these outcomes public, all parties can compute, 
\begin{eqnarray}
 D = t_1 \oplus_d t_2 \oplus_d\dots \oplus t_5,  
\end{eqnarray}
where, $\oplus_d$ is addition modulo $d$. Note that $D\in[0,d-1]$. Now let we use the states for which $D\in [0,k-1]$ for self-testing and those with $D\in [k,d-1]$ for MAKD. This corresponds to a fair and random $k/d$ vs $1-k/d$ partitioning of the shared states. An appropriate choice of $f$ may depend on the choice of self-testing inequalities and the total number of shared qubits. Note that the partitioning between self-test states and MAKD states is performed after the qubits are already out of the quantum channels and in the possession of each party. Therefore, the public communications to perform this partitioning cannot be useful for a potential eavesdropper. Note that we have assumed that we have access to authenticated classical channels, i.e., classical channels where the identity of each transmitter is known with certainty.
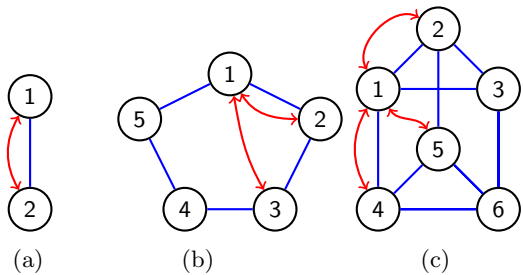
\begin{figure}[ht]
    \centering
    \begin{subfigure}[t]{0.12\textwidth}
        \begin{tikzpicture}[thick,main node/.style=     {circle,draw,font=\sffamily\small},scale=0.5, every node/.style={scale=1}]
            \node[main node](1) at (0,3) {1};
            \node[main node](2) at  (0,0) {2};
            \draw[thick,blue] (1) -- (2);
            \draw [<->,red] (1) to [out = 240,in = 120] (2);
        \end{tikzpicture}
        \subcaption{}
    \end{subfigure}
    \begin{subfigure}[t]{0.12\textwidth}
        \begin{tikzpicture}[thick,main node/.style={circle,draw,font=\sffamily\small},scale=0.4, every node/.style={scale=1}]
            \node[main node](4) at (0,0) {4};
            \node[main node](3) at  (3,0) {3};
            \node[main node](2) at  (4.5,3) {2};
            \node[main node](1) at  (1.5,4.5) {1};
            \node[main node](5) at  (-1.5,3) {5};
            \draw[thick,blue] (1) -- (2) -- (3) -- (4) -- (5) -- (1);
            \draw [<->,red] (1) to [out = -60,in = 180] (2);
            \draw [<->,red] (1) to [out = -80 ,in=120] (3);
        \end{tikzpicture}
        \subcaption{}
        \end{subfigure} 
    \begin{subfigure}[t]{0.22\textwidth}
        \begin{tikzpicture}[thick,main node/.style={circle,draw,font=\sffamily\small},scale=0.8, every node/.style={scale=1}]
            \node[main node](1) at (0,2) {1};
            \node[main node](2) at  (1,3) {2};
            \node[main node](3) at  (2,2) {3};
            \node[main node](4) at  (0,0) {4};
            \node[main node](5) at  (1,1) {5};
            \node[main node](6) at  (2,0) {6};
            \draw[thick,blue] (1)--(2)--(3)--(1)--(4)--(5)--(6)--(4)--(6)--(3)--(6)--(5)--(2);
            \draw [<->,red] (1) to [out = 120 ,in = 150] (2);    
            \draw [<->,red] (1) to [out = 240 ,in = 120] (4);
            \draw [<->,red] (1) to [out = -60,in = 120] (5);
        \end{tikzpicture}
        \subcaption{}
    \end{subfigure}
    \caption{The blue lines show the underlying mathematical graph for the graph state and the red double headed arrows point to distinct pairs of communicants. Next to each intent we give a stabilizer element that can be used.(a) $1\leftrightarrow 2 : X_1 Z_2$. (b) $1 \leftrightarrow 2,3 : Z_1 X_2 Z_3$. (c) $1 \leftrightarrow 2,4 : Z_4 X_1 Z_2 Z_3$ and $1\leftrightarrow 5 : S_1 S_5 = X_1 Z_3 X_5 Z_6$.}
    \label{2,3-uniform}
\end{figure}
See figure \ref{2,3-uniform} for examples of MAKD using $n=2,5,6$ graph, AME states. $n=2$ case is trivial, in the sense of not requiring any cooperation and is identical to E-$91$ using $\ket{b}$. We use the notation $(i\leftrightarrow j)$ throughout this paper to represent QKD intent between parties $i$ and $j$. 

\section{Stabilizers and entanglement}
We have used perfect correlations in stabilizer states for QKD. AME simply places a restriction on the minimum support of every stabilizer, thereby necessitating cooperation. We now remove this very strong constraint of absolutely maximal entanglement.

Let us discuss what mathematical constraint besides stabilizer structure is useful for entanglement based QKD. In stabilizer states, all nonzero correlations are perfect. In a shared stabilizer state, to perform QKD between qubits/parties $(i,j)$, one can naively use any stabilizer element, $\sigma$, with support on $i$ and $j$ and all other qubits / parties in support of $\sigma$ are required to cooperate with measurements and public announcements of their measurement outcomes. For any choice of two qubits, one can find many stabilizer elements whose support contains $i$ and $j$. For graph states, an obvious choice is $\sigma = S_i \cdot S_j$. We describe two related scenarios below to illustrate why some choices of stabilizer elements are not suitable for QKD. 
\begin{enumerate}
    \item Consider the two states,
    \begin{eqnarray*}
        \ket{p} = \ket{00} \;\;;\;\; \ket{\phi_+} =(1/\sqrt{2})(\ket{00}+\ket{11}).     
    \end{eqnarray*}
    Note that,
    \begin{eqnarray*}
        \bra{p}Z_1 Z_2 \ket{p} = \bra{\phi_+}Z_1 Z_2 \ket{\phi_+} = +1.
    \end{eqnarray*}
    We clarify that the public is always aware of the identity of the shared state and the measurements performed by each party. Despite this knowledge, $\ket{\phi_+}$ can be used for QKD and $\ket{p}$ cannot. This difference arises from how,
\begin{eqnarray*}
    \bra{p} Z_{1\text{ or } 2} \ket{p} = +1\;\;;\;\; \bra{\phi_+} Z_{1\text{ or } 2}\ket{\phi_+} = 0. 
\end{eqnarray*}
That is, the ``key" generated by $\ket{00}$ with $Z_1, Z_2$ measurements by parties $1$ and $2$ is trivially known to the public, since it is a constant value. Whereas, $Z_1,Z_2$ measurements on $\ket{\phi_+}$ result in perfectly random outcomes only known to the two parties $1$ and $2$.
\item Consider the graph state $\ket{L_4}$ in figure \ref{l4}.
 \begin{figure}[h!]
    \centering
    \begin{tikzpicture}[thick,main node/.style={circle,draw,font=\sffamily\small},scale=1, every node/.style={scale=1}]
        \node[main node](1) at (0,1) {1};
        \node[main node](2) at  (1,1) {2};
        \node[main node](3) at  (2,1) {3};
        \node[main node](4) at  (3,1) {4};
        \draw[thick,blue] (1) -- (2) -- (3) -- (4);
    \end{tikzpicture}
          \caption{$\ket{L_4}$}
          \label{l4}
\end{figure}
To perform $1\leftrightarrow 4$, use of $\sigma = S_1 \cdot S_4 = X_1 Z_2 Z_3 X_4$ is not secure. When the measurement outcomes $Z_2$ and $Z_3$ are made public, a public party can consider the stabilizer elements $S_1 = X_1 Z_2$ and $S_4 = Z_3 X_4$ to compute the key. 
\end{enumerate}
After considering these two scenarios, we can state the following necessary condition regarding stabilizer choice for QKD,

\begin{fact}\label{substring condition}
    Secure QKD between parties $i$ and $j$ of a shared stabilizer state requires the use of a stabilizer element with support on $i$ and $j$ such that no substring of the stabilizer with support on $i$ or $j$ is a stabilizer element.
\end{fact}
One is unable to use the stabilizer $Z_1 Z_2$ with state $\ket{p}$ because $Z_1$ and $Z_2$ which are substrings of $Z_1 Z_2$ are also stabilizers with support on $1$ and $2$. Similarly, $X_1 Z_2 Z_3 X_4$ is not suitable for $1\leftrightarrow 4$ because $X_1 Z_2$ and $Z_3 X_4$ which are substrings of $X_1 Z_2 Z_3 X_4$ are also stabilizer elements with support on $1$ and $4$.

If the stabilizer used for QKD satisfies fact \ref{substring condition}, then the public results of the measurements do not fully determine the key. It may still be possible for other parties that share the same state to determine the key using public outcomes in combination with their own local measurement outcome. Arguably, this is not as unsecure a scenario as when the public outcomes ``leak" the key to nonsharers of the graph state. We consider the distinction between these two levels of security in section \ref{multiple keys section}. In scenario $1$, $\ket{p}$ or more generally for any separable state, it is not possible to find a stabilizer for QKD which satisfies fact \ref{substring condition}. Note that in scenario $2$ we see that a stabilizer element constructed from a product of overlapping $S_i$'s for the graph state is a secure choice. A proof of this can be found within the proof for Theorem \ref{main_result}, Lemma \ref{no substring using path}.

It turns out that a lack of separability is not just necessary but also sufficient for QKD between two qubits of a shared stabilizer state. Quite simply,

\begin{theorem} \label{main_result}
    To perform QKD between any two qubits of a shared stabilizer state, it is necessary and sufficient that the state is not separable such that the two qubits belong to different bipartitions.
\end{theorem} 

A proof of Theorem \ref{main_result} is provided in Appendix \ref{necessary and sufficient QKD}. Intuitively, one expects an entanglement-based QKD protocol to necessarily require a lack of separability. But the degree of bipartite entanglement in generic multi-qubit states, when quantified, can take a continuous range of values that can be arbitrary close to zero while still remaining non-separable. That is, a simple lack of separability may not be very useful for many applications. However, for stabilizer states, the entanglement across different bipartitions always comes in discrete ``chunks". Therefore, a lack of separability is equivalent to the presence of sufficient entanglement required for QKD in shared stabilizer states.

Theorem \ref{main_result} is the reason why we did not attempt to generalize MAKD to a protocol that uses $k$-uniform states\cite{Arnaud_2013,PhysRevA.90.022316,sudevan2022nqubit,shi2020constructions}. $k$-uniform states are $n$-qubit states whose every $k$-qubit reduction is maximally mixed. Naturally, AME states are a special case of $k$-uniformity, where $k=\lfloor n/2 \rfloor$. However, $k$-uniform states are not always useful for the type of QKD discussed in this work, since they can be separable\footnote{It can be shown that if $\ket{a}$ and $\ket{b}$ are $k$-uniform states, then $\ket{a}\otimes \ket{b}$ is also a $k$-uniform state.}.

The proof of theorem \ref{main_result} shows that any shared graph state that allows QKD between any two of its parties will be one constructed from a connected graph. Connected graphs can be defined as graphs for which there is a path between any two qubits $i$ and $j$. In fact, as the proof discusses, any local unitary equivalent graph state constructed from a shared stabilizer state that allows secure QKD between any two of its qubits or parties will correspond to a connected graph. To perform QKD between any two qubits, a stabilizer element consisting of overlapping stabilizer elements along the path is a secure choice. Since there can be multiple possible paths and choice of overlapping stabilizer elements, the stabilizer choice for QKD between some specific pair ($i,j$) is not unique. In a practical setting, a table of stabilizer elements, one for each pair of communicants, can be decided in advance. We provide an example of a connected graph state and an associated communicant pair - stabilizer table in figure \ref{graph for table}.

\begin{figure}[ht]
\centering
\begin{subfigure}[t]{0.45\textwidth}
\begin{tikzpicture}[thick,main node/.style={circle,draw,font=\sffamily\small}]
    \node[main node](1) at (0,1) {1};
    \node[main node](2) at  (1,1) {2};
    \node[main node](3) at  (1,0) {3};
    \node[main node](4) at  (2,1) {4};
    \node[main node](5) at  (2,0) {5};
    \draw[thick,blue] (1) -- (2) -- (4) -- (5);
    \draw[thick,blue] (2) -- (3);
\end{tikzpicture}    
\end{subfigure}
\begin{subfigure}[t]{0.45\textwidth}

\begin{tabular}{ c|| c | c | c | c | c |}
$i / j$ & $1$ & $2$ & $3$ & $4$ & $5$ \\
\hline\hline
$1$ & \textunderscore & $X_1 Z_2$ & $X_1 X_3$ & $X_1 X_4 Z_5$ & $X_1 X_4 Z_5$ \\ 
\hline
$2$ & $X_1 Z_2$ & \textunderscore & $Z_2 X_3$ & $Z_2 X_4 Z_5$ & $Z_2 X_4 Z_5$ \\ 
\hline
$3$ & $X_1 X_3$ & $Z_2 X_3$ & \textunderscore & $X_3 X_4 Z_5$ & $X_3 X_4 Z_5$ \\ 
\hline
$4$ & $X_1 X_4 Z_5$ & $Z_2 X_4 Z_5$ & $X_3 X_4 Z_5$ & \textunderscore & $Z_4 X_5$ \\ 
\hline
$5$ & $X_1 X_4 Z_5$ & $Z_2 X_4 Z_5$ & $X_3 X_4 Z_5$ & $Z_4 X_5$ & \textunderscore  \\
\hline
\end{tabular}

\end{subfigure}
    \caption{ Table of stabilizer correlations for all $i,j$ communicant pairs for the above graph.}
    \label{graph for table}
\end{figure}

Now we can summarize the steps to perform QKD between qubits $i$ and $j$ of an $n$-qubit graph state shared between $n$ parties,

\begin{itemize}
    \item \textbf{Step 0: } Construct a pair stabilizer table for the resource graph state. This step is performed once; all subsequent QKD's continue to use the same table.
    \item \textbf{Step 1: }Party $i$ publicly announces its intention to share a key with party $j$.
    \item \textbf{Step 2: } All cooperators according to the pair stabilizer table make the corresponding Pauli measurements and publicly announce their measurement outcomes.
    \item \textbf{Step 3: } If all necessary permission bits are public, $i,j$ can use their corresponding stabilizer correlation and their own measurement outcomes to compute the key.
\end{itemize}

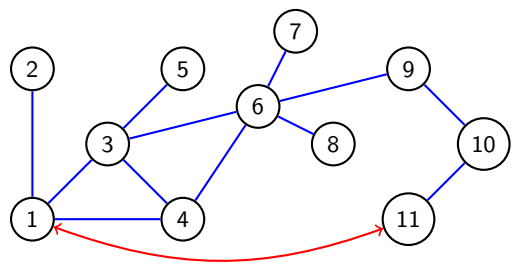
\begin{figure}[ht]
\centering
\begin{tikzpicture}[thick,main node/.style={circle,draw,font=\sffamily\small}]
    \node[main node](1) at (0,0) {1};
    \node[main node](2) at  (0,2) {2};
    \node[main node](3) at  (1,1) {3};
    \node[main node](4) at  (2,0) {4};
    \node[main node](5) at  (2,2) {5};
    \node[main node](6) at  (3,1.5) {6};
    \node[main node](7) at  (3.5,2.5) {7};
    \node[main node](8) at  (4,1) {8};
    \node[main node](9) at  (5,2) {9};
    \node[main node](10) at  (6,1) {10};
    \node[main node](11) at  (5,0) {11};
    \draw[thick,blue] (1) -- (2); \draw[thick,blue] (1) -- (4);\draw[thick,blue] (1) -- (3); \draw[thick,blue] (3) -- (4); \draw[thick,blue] (3) -- (5); \draw[thick,blue] (3) -- (6); \draw[thick,blue] (6) -- (7); \draw[thick,blue] (6) -- (8); \draw[thick,blue] (6) -- (9); \draw[thick,blue] (9) -- (10); \draw[thick,blue] (10) -- (11); \draw[thick,blue] (4) -- (6);
    \draw [<->,red] (1) to [out = -20 ,in = 200] (11);
\end{tikzpicture}
    \caption{For $1\leftrightarrow 11$, we choose the path, $p=\{1,3,6,9,10,11\}$ and use the stabilizer element $S_1 S_6 S_{10} = X_1 Z_2 X_6 Z_7 Z_8 X_{10} Z_{11}$.}
    \label{graph path}
\end{figure}

\section{More than one key per state} \label{multiple keys section}
Note that QKD between parties that are far apart with respect to the underlying graph for the shared graph state will require many measurement outcomes to be made public. Using a product of overlapping stabilizers along a path connecting the two parties only ensures that the public measurement outcomes does not reveal the key to the general public. It is possible that the public outcomes allow some other party possessing a qubit of the shared graph state to determine the secret key. In figure \ref{graph path}, note that the $X_{10}$ public outcome can be used for QKD between parties $9$ and $11$ using $Z_9 X_{10} Z_{11}$. Even though the intended communication is between parties $1$ and $11$, it becomes possible for $9$ to also determine the key using only their local qubit and public information. This can be fixed by changing the stabilizer for QKD to,
\begin{eqnarray}
    \sigma = \prod_{k\in P(i,j)} S_k,
\end{eqnarray}
where, $P(i,j)$ is the set of all vertices along some path connecting $i$ with $j$. On the other hand, constructing $\sigma$ from a subset of the vertices in $P(i,j)$ with overlapping stabilizers can be used to realize conference keys. By conference keys\cite{1056542,Murta_2020,Pickston2023,Epping_2017,fedrizzi2022quantum}, we are referring to the situation where more than two parties share the same key. In figure \ref{graph path} we have a common key between parties possessing qubits $1,9$ and $11$. A more obvious example of conference keys comes from the generalized GHZ state,
\begin{eqnarray*}
    \ket{GHZ} = \frac{1}{\sqrt{2}}\Big(\ket{0}^{\otimes n} + \ket{1}^{\otimes n}\Big),
\end{eqnarray*}

where the conference keys result from how, $\forall i,j\in [1,2,3,\dots,n]$, $\langle Z_i Z_j \rangle = +1$.   

Since now we are considering multiple QKD pairs and stabilizers per state, we need to formalize these notions. 
Let $\{c_1,c_2,c_3,\dots c_r\}$ be the communicant pairs and their corresponding stabilizer elements used for QKD be $\{\sigma_1,\sigma_2,\sigma_3,\dots\sigma_r\}$. $\sigma_i \backslash c_i$ is the support of $\sigma_i$ besides the two communicants. The measurement outcomes in $\sigma_i \backslash c_i$ is public knowledge.

We list the constraints on the choice of simultaneous stabilizer correlations below\label{simultaneous stabilizer constraint},

\begin{itemize}
    \item \textbf{Constraint 1: } No qubit in any $c_i$ should be in any $\sigma_j \backslash c_j, j\neq i$.  That is, a qubit whose measurement result is meant to be a key should not, also be simultaneously necessary as a permission bit for the distribution of a different key.
    \item \textbf{Constraint 2: }$\forall i,j$ let $A= \text{support}(\sigma_i)\cap \text{support}(\sigma_j)$, then, substring of $\sigma_i$ and $\sigma_j$ on $A$ should be identical. That is, no qubit or party should be required to perform different Pauli measurements simultaneously .
\end{itemize}

Let $\sigma_{i,j}$ and $\sigma_{i,k}$ be two stabilizer elements for QKD between parties $i,j$ and $i,k$ respectively, that satisfy constraints $1$ and $2$, then the stabilizer element $\sigma_{i,j}\cdot \sigma_{i,k}$ is a secure stabilizer for QKD between parties $j,k$. This follows from how $\sigma_{i,j}\cdot \sigma_{i,k}$ has support on $j,k$ and is a product of two overlapping stabilizers whose security we have already proven in the proof of Theorem \ref{main_result}. Now we can state precisely what we refer to as conference key distribution.

\begin{fact}
  In a shared graph state, let a set of communicant pairs form the edges of a connected graph where the communicants belonging to these pairs are the vertices and we use stabilizer elements satisfying constraints $1$ and $2$ for each QKD pair, then all the communicants belonging to these pairs share a common/conference key.  
\end{fact}

Consider the graph state in figure \ref{conference key}. Let $1\leftrightarrow 2$ be accomplished with $S_1 = Z_5 X_1 Z_2$ and $1\leftrightarrow 3$ with $S_1 \cdot S_3 = X_1 Z_5 Z_4 X_3$. The public outcome of $Z_4$ results in $2\leftrightarrow 3$ also being possible using $S_3 = Z_2 X_3 Z_4$.   We list all distinct three party conference key scenarios using this $5$-qubit state in figure \ref{conference key}. 

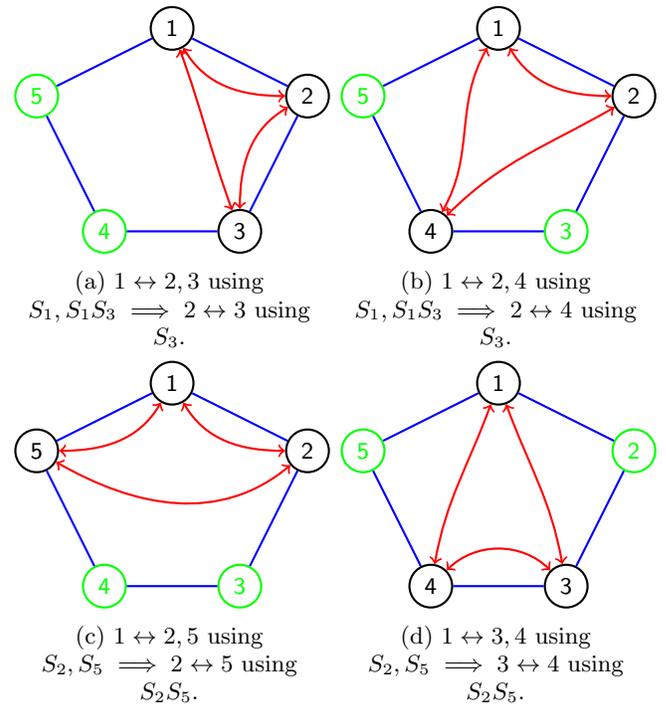
\begin{figure}[ht]
    \centering
    \begin{subfigure}[]{0.23\textwidth}
        
        \begin{tikzpicture}[thick,main node/.style={circle,draw,font=\sffamily\small},scale=0.6, every node/.style={scale=1}]
    \node[main node,green](4) at (0,0) {4};
    \node[main node](3) at  (3,0) {3};
    \node[main node](2) at  (4.5,3) {2};
    \node[main node](1) at  (1.5,4.5) {1};
    \node[main node,green](5) at  (-1.5,3) {5};
    \draw[thick,blue] (1) -- (2) -- (3) -- (4) -- (5) -- (1);
    \draw [<->,red] (1) to [out = -60,in = 180] (2);
    \draw [<->,red] (1) to [out = -70 ,in = 110] (3);
    \draw [<->,red] (2) to [out = 210 ,in = 90] (3);
\end{tikzpicture}
        \subcaption{$1\leftrightarrow 2,3$ using $S_1, S_1 S_3 \implies 2\leftrightarrow 3$ using $S_3$.}
    \end{subfigure}%
    ~ 
    \begin{subfigure}[]{0.23\textwidth}
        
        \begin{tikzpicture}[thick,main node/.style={circle,draw,font=\sffamily\small},scale=0.6, every node/.style={scale=1}]
    \node[main node](4) at (0,0) {4};
    \node[main node,green](3) at  (3,0) {3};
    \node[main node](2) at  (4.5,3) {2};
    \node[main node](1) at  (1.5,4.5) {1};
    \node[main node,green](5) at  (-1.5,3) {5};
    \draw[thick,blue] (1) -- (2) -- (3) -- (4) -- (5) -- (1);
    \draw [<->,red] (1) to [out = -60,in = 180] (2);
    \draw [<->,red] (1) to [out = -120 ,in = 60] (4);
    \draw [<->,red] (4) to [out = 45 ,in = 210] (2);
\end{tikzpicture}
        \subcaption{$1\leftrightarrow 2,4$ using $S_1, S_1 S_3 \implies 2\leftrightarrow 4$ using $S_3$.}
    \end{subfigure}
    \begin{subfigure}[]{0.23\textwidth}
        
        \begin{tikzpicture}[thick,main node/.style={circle,draw,font=\sffamily\small},scale=0.6, every node/.style={scale=1}]
    \node[main node,green](4) at (0,0) {4};
    \node[main node,green](3) at  (3,0) {3};
    \node[main node](2) at  (4.5,3) {2};
    \node[main node](1) at  (1.5,4.5) {1};
    \node[main node](5) at  (-1.5,3) {5};
    \draw[thick,blue] (1) -- (2) -- (3) -- (4) -- (5) -- (1);
    \draw [<->,red] (1) to [out = -60,in = 180] (2);
    \draw [<->,red] (1) to [out = -120 ,in = 0] (5);
    \draw [<->,red] (5) to [out = -30 ,in = 220] (2);
\end{tikzpicture}
        \subcaption{$1\leftrightarrow 2,5$ using $S_2, S_5 \implies 2\leftrightarrow 5$ using $S_2 S_5$.}
    \end{subfigure}%
    ~ 
    \begin{subfigure}[]{0.23\textwidth}
        
        \begin{tikzpicture}[thick,main node/.style={circle,draw,font=\sffamily\small},scale=0.6, every node/.style={scale=1}]
    \node[main node](4) at (0,0) {4};
    \node[main node](3) at  (3,0) {3};
    \node[main node,green](2) at  (4.5,3) {2};
    \node[main node](1) at  (1.5,4.5) {1};
    \node[main node,green](5) at  (-1.5,3) {5};
    \draw[thick,blue] (1) -- (2) -- (3) -- (4) -- (5) -- (1);
    \draw [<->,red] (1) to [out = -70,in = 100] (3);
    \draw [<->,red] (1) to [out = -110 ,in = 80] (4);
    \draw [<->,red] (4) to [out = 45 ,in = 135] (3);
\end{tikzpicture}
        \subcaption{$1\leftrightarrow 3,4$ using $S_2, S_5 \implies 3\leftrightarrow 4$ using $S_2 S_5$.}
    \end{subfigure}
    
    \caption{ The red double headed arrows represent communicants with shared keys. Qubits whose measurement results are public are green.}
    \label{conference key}
\end{figure}

In figure \ref{implicit conference keys long chain}, let parties $1$ and $9$ attempt QKD using the stabilizer element $S_2 \cdot S_4 \cdot S_6 \cdot S_8$. This requires $X_2, X_4, X_6$ and $X_8$ measurement outcomes being made public. Since parties $3,5,7$ were not required to participate in QKD between $1$ and $9$, they remain free to use the stabilizer correlations $S_2$ for $1\leftrightarrow 3$, $S_4$ for $3\leftrightarrow 5$, $S_6$ for $5\leftrightarrow 7$ and $S_8$ for $7\leftrightarrow 9$ while satisfying constraints $1$ and $2$ and therefore set up a conference QKD between parties $\{1,3,5,7,9\}$ since now the QKD pairs form the edges of a connected graph. 

\begin{figure}[ht]
\centering

\begin{tikzpicture}[thick,main node/.style={circle,draw,font=\sffamily\small}]
    \node[main node](1) at (0,0) {1};
    \node[main node,green](2) at  (0.75,0) {2};
    \node[main node](3) at  (1.5,0) {3};
    \node[main node,green](4) at  (2.25,0) {4};
    \node[main node](5) at  (3,0) {5};
    \node[main node,green](6) at  (3.75,0) {6};
    \node[main node](7) at  (4.5,0) {7};
    \node[main node,green](8) at  (5.25,0) {8};
    \node[main node](9) at  (6,0) {9};
    \node[main node](10) at  (6.75,0) {10};
    \draw[thick,blue] (1)--(2)--(3)--(4)--(5)--(6)--(7)--(8)--(9)--(10);
    
    \draw[loosely dotted] (-1,0) -- (1);
    \draw[loosely dotted] (10) -- (7.75,0);
    \draw [<->,red] (1) to [out = 35,in = 145] (9);
    \draw [<->,red] (1) to [out = -60,in = 240] (3); \draw [<->,red] (3) to [out = -60,in = 240] (5); \draw [<->,red] (5) to [out = -60,in = 240] (7); \draw [<->,red] (7) to [out = -60,in = 240] (9);
    
\end{tikzpicture}
    \caption{Performing $1\leftrightarrow 9$ using $S_2 \cdot S_4 \cdot S_6 \cdot S_8$ results in enough public measurement outcomes to allow a conference key setup between all the qubits in $\{1,3,5,7,9\}$. }
    \label{implicit conference keys long chain}
\end{figure}
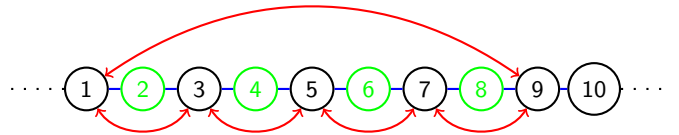
\begin{fact}
  If multiple non-overlapping pairs of communicants (the pairs of communicants do not form the vertices of a connected graph) perform QKD using stabilizers satisfying constraints $1$ and $2$, then we can distribute multiple, parallel or independent keys per resource state.  \end{fact}
  See figure \ref{Parallel keys} for an example.

\begin{figure}[ht]
    \centering
    \begin{subfigure}[t]{0.23\textwidth}
        
        \begin{tikzpicture}[thick,main node/.style={circle,draw,font=\sffamily\small\bfseries},scale=0.6, every node/.style={scale=1}]
    \node[main node](4) at (0,0) {4};
    \node[main node,green](3) at  (3,0) {3};
    \node[main node](2) at  (4.5,3) {2};
    \node[main node](1) at  (1.5,4.5) {1};
    \node[main node](5) at  (-1.5,3) {5};
    \draw[thick,blue] (1) -- (2) -- (3) -- (4) -- (5) -- (1);
    \draw [<->,red] (1) to [out = -60,in = 180] (2);
    \draw [<->,red] (5) to [out = 0 ,in = 60] (4);
\end{tikzpicture}
    \end{subfigure}%
    ~ 
    \begin{subfigure}[t]{0.23\textwidth}
        
        \begin{tikzpicture}[thick,main node/.style={circle,draw,font=\sffamily\small\bfseries},scale=0.6, every node/.style={scale=1}]
    \node[main node](4) at (0,0) {4};
    \node[main node](3) at  (3,0) {3};
    \node[main node,green](2) at  (4.5,3) {2};
    \node[main node](1) at  (1.5,4.5) {1};
    \node[main node](5) at  (-1.5,3) {5};
    \draw[thick,blue] (1) -- (2) -- (3) -- (4) -- (5) -- (1);
    \draw [<->,red] (1) to [out = -80,in = 120] (3);
    \draw [<->,red] (5) to [out = 0 ,in = 60] (4);
\end{tikzpicture}
    \end{subfigure}
    \caption{(a) $1\leftrightarrow2$ using $Z_1 X_2 Z_3$, $4\leftrightarrow5$ using $Z_5 X_4 Z_3$. (b) $1\leftrightarrow3$ using $Z_1 X_2 Z_3$, $4\leftrightarrow5$ using $S_2 \cdot S_4 \cdot S_5=X_2 Y_4 Y_5$. Note that a common public measurement outcome allows both independent QKD events.}
    \label{Parallel keys}
\end{figure}
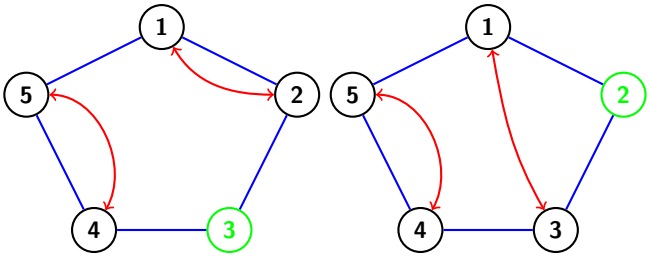

\textit{Fixed set of ``authorizers" for QKD across some bipartition.\textemdash}The fact that QKD using shared graph states requires the use of ``overlapping stabilizers" lends importance to the topology of the underlying graph. If one uses a long chain of qubits to entangle a graph state across a separable bipartition, the parties along this long chain become necessary authorizers for any communication across this bipartition as shown in figure \ref{long chain auth}. In fact, increasing the number of parties along this chain increases the number of permission bits needed for QKD between the parties in $G_1$ and $G_2$. 

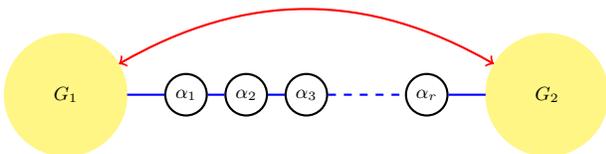
\begin{figure}[ht]
\centering
\begin{tikzpicture}[thick,main node/.style={circle,draw,font=\sffamily\small\bfseries},scale=0.8, every node/.style={scale=0.8},roundnode/.style={circle, draw=yellow!60, fill=yellow!60, very thick, minimum size=2 cm}]
    \node[roundnode](1) at (0,0) {$G_1$};
    \node[roundnode](2) at (8,0) {$G_2$};
    \node[main node](3) at  (2,0) {$\alpha_{1}$};
    \node[main node](4) at  (3,0) {$\alpha_{2}$};
    \node[main node](5) at  (4,0) {$\alpha_{3}$};
    \node[main node](6) at  (6,0) {$\alpha_{r}$};
    \draw[thick,blue] (1)--(3)--(4)--(5);
    \draw[dashed,blue] (5)--(6);
    \draw[thick,blue] (6)--(2);
    \draw [<->,red] (1) to [out = 30 ,in = 150] (2);
\end{tikzpicture}
    \caption{QKD between any party in graph $G_1$ and a party in graph $G_2$ requires the cooperation of at least $\lfloor r/2 \rfloor$ parties in $\{ \alpha_1, \alpha_2, \alpha_3, \dots, \alpha_r \}$ such that their corresponding stabilizers overlap.}
    \label{long chain auth}
\end{figure}

\section{Summary and outlook}
We have proven how shared $n$-qudit AME stabilizer states allow for QKD protocols that require the cooperation of a minimum number ($\lfloor n/2\rfloor-1$) of other parties sharing the state. AME structure makes this minimum number necessary, and the stabilizer structure makes this minimum number sufficient. One can explore if AME structure alone is sufficient for MAKD. Either all AME qudit states have a stabilizer structure up to some choice of local basis, in which case AME structure alone allows MAKD. On the other hand, if there exists AME states that are not LU equivalent to any stabilizer state, then one can explore if QKD can still be accomplished using the partial correlations between local observables. 

In this work, the self-testing section simply borrowed a Bell type inequality from \cite{PhysRevLett.124.020402}. The self-testing we have used here does not explicitly take the AME structure into account. Since it is known that the ability to perform local self-testing of shared states is dependent on the entanglement content of the state\cite{Makuta_2021,_upi__2018}, there may exist more efficient protocols\cite{sarkar2021self,jebarathinam2019maximal,santos2023scalable} for self-testing AME states in particular. In addition, we are yet to perform a careful analysis of the key rates\cite{Grasselli_2018,oslovich2024efficientmultipartyquantumkey}. 

We have illustrated with three variants(conference keys, independent keys and being able to have some fixed set of ``authorizers" for QKD between two bipartitions) how the underlying topology of the graph state can be used to realize various constrained QKD protocols. One can attempt to find many more customized applications for entangled stabilizer/graph states. Considering the rapid progress in experiments to construct graph states\cite{Huang2023,Thomas_2022,Lu_2007,huang2023chip,hoyer2006resources,gnatenko2021entanglement,mooney2021generation,li2020multiphoton,teper2020thermal,Cooper_2024,Xu_2021} and demonstration of various protocols\cite{Bell_2014,lu2024quantum,joo2013generating,markham2020simple,zander2024benchmarking,Lee_2012} using them, our findings in this work appear feasible to implement in the near term.   

\textit{ Acknowledgements.\textemdash} We acknowledge several encouraging discussions with Guruprasad Kar. S.S acknowledges Manisha Goyal for help with the proof of Lemma 1 and Ayan Biswas for performing some preliminary numerical checks of the protocol\cite{qiskit_noise_model,Hunter:2007}. S.S acknowledges the fellowship support of IISER-Kolkata. This research was supported in part by the International Centre for Theoretical Sciences (ICTS) for participating in the program Condensed Matter Meets Quantum Information (ICTS/COMQUI2023/09).

\twocolumngrid

\bibliography{references}
\appendix
\section{Proof of Theorem \ref{AME theorem}}\label{proof theorem 1}
\begin{proof}
Let $\ket{\psi}$ be an $n$-qubit, AME, stabilizer state. Let $\mathcal{S}$ be the stabilizer group of $\ket{\psi}$. We have,
\begin{eqnarray}
    \ket{\psi}\bra{\psi} \equiv \rho = \frac{1}{2^n}\sum_{\sigma\in \mathcal{S}}\sigma.
\end{eqnarray}
The reduced density matrix of a subset of qubits $A$ is given by,
\begin{eqnarray}
    \rho_A = \frac{1}{2^{|A|}} \sum_{\sigma \in \mathcal{S}_A} \sigma = \frac{1}{2^{|A|}}\prod_{\sigma\in\{S_i^A\}}\Big(I+\sigma\Big),
\end{eqnarray}
where, $\mathcal{S}_A$ represents the subgroup of $\mathcal{S}$ with support only on the set of qubits $A$ and $\{S_i^A\}$ are the generators of $\mathcal{S}_A$. ``$I$" refers to the identity matrix whose dimensions are clear from the context in which it is used. 

Since all elements of the stabilizer group commute and they all square to identity, the number of elements of a stabilizer group for qubits is always some power of two. This applies to it's subgroups as well. Let $|\mathcal{S_A}| = 2^k \implies |\{S_i^A\}|= k$. Now,
\begin{eqnarray}
    \rho_A = \frac{1}{2^{|A|-k}}\prod_{\sigma\in \{S_i^A\}}\frac{I+\sigma}{2}.
\end{eqnarray}
Note that,
\begin{eqnarray*}
    \Big(\prod_{\sigma\in \{S_i^A\}}\frac{I+\sigma}{2}\Big)^2 = \prod_{\sigma\in \{S_i^A\}}\frac{I+\sigma}{2}. 
\end{eqnarray*}
Therefore, the eigenvalues of $\prod_{\sigma\in \{S_i^A\}}\frac{I+\sigma}{2}$ are $\{0,1\}$. Let the multiplicity of the ``$1$" eigenvalue be $r$. Since $\Tr{\rho_A} = 1$, therefore,
\begin{align*}
    \frac{r}{2^{|A|-k}} = 1 \implies r = 2^{|A|-k}.
\end{align*}
Recall that $\ket{\psi}$ is an AME state. Therefore,
\begin{align*}
    \rho_A = \frac{I}{2^{|A|}} \;;\; |A|\in [1,\lfloor n/2\rfloor].
\end{align*}
Since entanglement entropy is symmetric across a bipartition, 
\begin{align*}
    E(\rho_A) = E(\rho_B),
\end{align*}
Where $E(\rho)$ is the Von Neumann entropy of $\rho$, $A$ and $B$ are disjoint bipartitions of all qubits in $\ket{\psi}$.
Consider an arbitrary subset of qubits $A$, such that $|A|=\lfloor n/2 \rfloor +1$. Then, the complementary set $B$ of qubits has $|B|=n- \lfloor n/2\rfloor - 1$. Therefore, $E(\rho_A) = E(\rho_B)$ implies,
\begin{align*}
  E\Big(\frac{1}{2^{|A|-k}}\prod_{\sigma\in \{S_i^A\}}\frac{I+\sigma}{2}\Big) = E(\frac{I}{2^{n-\lfloor n/2 \rfloor-1}}),  
\end{align*}
\begin{align*}
    \implies \lfloor n/2 \rfloor + 1 - k = n-\lfloor n/2 \rfloor -1.\\
    \implies k = \begin{cases}
        1 \;\;\text{when $n$ is odd} \\
        2 \;\; \text{when $n$ is even}.
    \end{cases}
\end{align*}
Since the number of stabilizer elements with support only on $A$, $|\mathcal{S}_A| = 2^k$, and since this includes the identity matrix, we have $2^1 -1 = 1$ and $2^2-1 =3$ stabilizer choices for QKD between two qubits in $A$ for odd and even $n$, respectively. 
\end{proof}

\section{Proof that the public bits in MAKD do not allow any unintended parties to compute the key}\label{substring AME}
Consider MAKD between parties $i$ and $j$ using the stabilizer element $\sigma_i \otimes \sigma_j^{\prime} \otimes \sigma_{A}^{\prime\prime}$, where $\sigma_i,\sigma_j^{\prime}$ and $\sigma_{A}^{\prime\prime}$ are Pauli group elements on qubits $i,j$ and set of qubits $A$ respectively. The weight of $\sigma_i \otimes \sigma_j^{\prime} \otimes \sigma_{A}^{\prime\prime}$ is $\lfloor n/2\rfloor + 1$.  Note that the set of qubits $A$ are with the parties that are only involved in permitting MAKD. Therefore, the measurement outcomes of the local operators in $\sigma_A^{\prime\prime}$ are known to the public.
\begin{itemize}
    \item For a member of the public to be capable of computing the key, either $\sigma_i\otimes \sigma_{A}^{\prime\prime}$ or $\sigma_j\otimes \sigma_{A}^{\prime\prime}$ should belong to the stabilizer group of the state. But, this operator has weight $\lfloor n/2\rfloor$ which is less than the weight of the smallest stabilizer in the shared AME state ($\lfloor n/2\rfloor+1$). Therefore, the public cannot compute the key.
    \item For a party $k\notin \{i,j\}\cup A$ to be capable of computing the key, either $\sigma_i\otimes \sigma_{A}^{\prime\prime}\otimes \mathcal{O}_k$ or $\sigma_j^{\prime}\otimes \sigma_{A}^{\prime\prime}\otimes \mathcal{O}_k$ should be a stabilizer element, where $\mathcal{O}_k$ is some observable on qubit $k$. The product of two stabilizer elements is also a stabilizer elements, therefore,
    \begin{eqnarray}
       \sigma_i \otimes \sigma_j^{\prime} \otimes \sigma_{A}^{\prime\prime} \cdot \sigma_i\otimes \sigma_{A}^{\prime\prime}\otimes \mathcal{O}_k = \sigma_j^{\prime} \otimes \mathcal{O}_k  
    \end{eqnarray}
    is a stabilizer element. This stabilizer element has weight of $2$, which is only possible for AME states with $n=2,3$, which are ``trivial" cases of MAKD, since extra cooperators are not necessary for QKD. Therefore, the other parties sharing the state cannot compute the key.
\end{itemize}

\section{Proof of Theorem \ref{qudit ame theorem}}\label{qudit ame proof}
\begin{proof}
We need a basis of matrices analogous to the Pauli group for qubits for the case of qudits. The density matrix of a single qudit can be expanded in a basis consisting of $d^2$ matrices (including the identity matrix) which are finitely generated by the following two traceless matrices,
\begin{eqnarray}
    X(d) = \sum_{l=0}^{d-1} \ket{l\oplus_d 1}\bra{l} \;\;\And\;\; Z(d) = \sum_{l=0}^{d-1}\omega^l \ket{l}\bra{l},
\end{eqnarray}
where, $\oplus_d$ is addition modulo $d$ and $\omega = \exp{2\pi i/d}$ i.e., the $d$'th root of unity. The $d^2$ generalized Paulis finitely generated by the above two matrices are given by,
\begin{eqnarray}
    D^{\mu,\nu} = \exp{i\pi\mu\nu/d}X(d)^\mu Z(d)^\nu.
\end{eqnarray}
An arbitrary $n$-qudit density matrix can be expanded in a basis of $d^{2n}$ tensor products of $D^{\mu,\nu}$'s which are elements of the qudit Pauli group $\mathcal{P}_d$. A stabilizer qudit state can be specified using $n$ commuting and independent elements of $\mathcal{P}_d$. Since $(D^{\mu,\nu})^d = I$, the stabilizer group consists of $d^n$ elements. AME qudit states can be defined in the same way as for qubits, i.e., states whose $\lfloor n/2 \rfloor$ reductions are all maximally mixed. Now consider an $n$-qudit state $\ket{\psi}$ with the stabilizer group $\mathcal{S}$ and some set of qubits $A$. Similarly to the proof of theorem \ref{AME theorem}, we use the equality of the entanglement entropy across a bipartition and the fact that the reduced density matrix of $A$ qudits can be written as a sum over the stabilizer elements with support only on $A$, to arrive at,
\begin{eqnarray}
    E\Big(\frac{I}{d^|A|}\Big) = E\Big(\frac{1}{d^{n-|A|}}\sum_{\sigma\in \mathcal{S}_{N\backslash A}} \sigma\Big), 
\end{eqnarray}
where, $\mathcal{S}_A$ is the subgroup of $\mathcal{S}$ with support only on $|A|$. Using $\log$ of base $d$ to compute entanglement in ``dits",
\begin{eqnarray}
    \lfloor n/2 \rfloor = E\Big(\frac{1}{d^{n-|A|}}\sum_{\sigma\in \mathcal{S}_{N\backslash A}} \sigma\Big).
\end{eqnarray}
Now, $\mathcal{S}_{N\backslash A}$ cannot be a trivial group ($\{I\}$) since $|A|<|N\backslash A|$, therefore it contains at least one stabilizer element.
\end{proof}

  \section{Proof of theorem \ref{main_result}}\label{necessary and sufficient QKD}

\begin{lemma}\label{separable stabilizer}

Given any separable stabilizer state $\ket{a}$ such that $\ket{a} =  \ket{b} \otimes \ket{c}$ implies that $\ket{b}$ and $\ket{c}$ are also stabilizer states. 

\end{lemma}

\begin{proof} We require the following definition of stabilizer states from \cite{Scott_qcbook_2018},
\begin{definition}
    $n$-qubit stabilizer states in the computational basis have $2^l$ non-zero coefficients of equal weight, where $l\in \mathbb{Z}_+$ and $l \leq n$.
\end{definition}
     We prove that $\ket{b}$ and $\ket{c}$ are also stabilizer states by showing that they satisfy the above definition when $\ket{a}$ is a stabilizer state.

\begin{sublemma}
    The number of non-zero coefficients in $\ket{b}$ and $\ket{c}$ are powers of $2$.
\end{sublemma}
\begin{proof}
 Let $\ket{a}$ be an $n$-qubit stabilizer state. Up to some normalization constant, all coefficients of $\ket{a}$ in the computational basis, $ a_i \in \{0,\pm 1\} $\cite{796376}. The number of non-zero $a_i$'s is $2^l$, where $l\leq n$. Let the number of non-zero coefficients in $\ket{b},\ket{c}$ be $\lambda$ and $\tau$. Since $\ket{a} = \ket{b}\otimes \ket{c}$, the number of non-zero coefficients are equal on both sides, and therefore, 
 \begin{eqnarray*}
  2^l = \lambda \times \tau,   
 \end{eqnarray*}
 therefore, $\lambda$ and $\tau$ are also powers of $2$.     
\end{proof}

\begin{sublemma}
    $\ket{b}$ and $\ket{c}$ can be written such that one of their coefficients each is $1$ and the overall normalization constants for both take the form $\frac{1}{2^{r/2}}$ and $\frac{1}{2^{s/2}}$ respectively, where $r,s \in \mathbb{Z}_+$.
\end{sublemma}
\begin{proof}
The normalization constant for $\ket{a}$ is $\frac{1}{2^{n/2}}$. Consider a non-zero coefficient of $\ket{a}$, $a_i$. Comparing $a_i$ with its corresponding element in $\ket{b}\otimes \ket{c}$ we have,
\begin{eqnarray*}
    \frac{1}{2^{n/2}} a_i = \frac{1}{\sqrt{N_b} \sqrt{N_c}}b_j c_k,
\end{eqnarray*}
where, $\frac{1}{\sqrt{N_b}}$ and $\frac{1}{\sqrt{N_c}}$ are normalization constants for $\ket{b}$ and $\ket{c}$. Without loss of generality, we can replace any one specific choice of $b_j$ and $c_k$ with $1$'s and adjust the normalization constant accordingly. The same ``normalization trick" can be used on the left hand side of the above equation to replace $a_i$ with $1$. Now, we have,
\begin{eqnarray*}
    \frac{1}{2^{n/2}} = \frac{1}{\sqrt{N_b N_c}}.
\end{eqnarray*}
Therefore, $N_b$ and $N_c$ are also positive integer powers of two.    
\end{proof}

\begin{sublemma}
 All non-zero elements of $\ket{b}$ and $\ket{c}$ are equal weight.   
\end{sublemma}
\begin{proof}
 Comparing the non-zero coefficients on both sides of $\ket{a} = \ket{b}\otimes \ket{c}$, we have,
\begin{eqnarray*}
    \pm = a_i b_j \;\; \text{and} \;\; \pm a_i b_k \implies b_j = \pm b_k.
\end{eqnarray*}
similarly we can show that all non-zero elements of $\ket{c}$ are equal weight.   
\end{proof}

From the sublemmas $1$,$2$ and $3$, which we have proved, it is clear that $\ket{b}$ and $\ket{c}$ satisfy the definition of stabilizer states.
     
\end{proof}

\begin{lemma} \label{necessary MPA-QKD}
    If a stabilizer state is separable across a bipartition $A$ and $B$ such that $i\in A$ and $j\in B$, then secure QKD is not possible between $i$ and $j$. 
\end{lemma}

\begin{proof}
    Let $\ket{\psi}$ be separable across the $A,B$ bipartition such that the two communicants ($i,j$) belong to the two disjoint sets of qubits $A$ and $B$,
    \begin{eqnarray*}
        \ket{\psi} =  \ket{\phi_A} \otimes \ket{\chi_B}.
    \end{eqnarray*}
As shown in lemma \ref{separable stabilizer}, $\ket{\psi}$ being a stabilizer state results in $\ket{\phi_A} $ and $\ket{\chi_B}$ also being stabilizer states. Let $\mathcal{S},\mathcal{S}(\phi)$ and $\mathcal{S}(\chi)$ be the stabilizer groups of $\ket{\psi}$, $\ket{\phi_A}$ and $\ket{\chi_B}$. The projector onto $\ket{\psi}$ can be expanded using its stabilizer group elements as given below,
\begin{eqnarray*}
    \ket{\psi}\bra{\psi} = \frac{1}{2^{|A|+|B|}}\sum_{\sigma\in \mathcal{S}} \sigma.
\end{eqnarray*}
Using $\ket{\psi}\bra{\psi} = \ket{\phi_A}\bra{\phi_A} \otimes \ket{\chi_B}\bra{\chi_B}$, and expanding the projectors on to $\ket{\phi_A}$ and $\ket{\chi_B}$ also in terms of their stabilizer group elements gives us,
\begin{eqnarray*}
  \frac{1}{2^{|A|}}\sum_{\sigma^1\in \mathcal{S}(\phi)} \sigma^1 \otimes \frac{1}{2^{|B|}}\sum_{\sigma^2\in \mathcal{S}(\chi)} \sigma^2 = \frac{1}{2^{|A|+|B|}}\sum_{\sigma\in \mathcal{S}} \sigma. 
\end{eqnarray*}

\begin{eqnarray*}
    \implies \sum_{\sigma^1 \in \mathcal{S}^1, \sigma^2 \in \mathcal{S}^2}\sigma^1 \otimes \sigma^2 = \sum_{\sigma\in \mathcal{S}} \sigma
\end{eqnarray*}

The linear independence of each $\sigma \in \mathcal{S}$ implies that every $\sigma \in \mathcal{S}$ can be written as a product of a stabilizer element of $\ket{\phi_A}$ and stabilizer element of $\ket{\chi_B}$. That is any stabilizer element with support on $i$ and $j$ that are separable will contain two substrings, each with support on $i$ and $j$, that are stabilizers/perfect correlations. As per the necessary condition that is fact \ref{substring condition}, such stabilizers reveal the key to all. Therefore MPA-QKD cannot be performed between separable qubits.

\end{proof}

 We have shown that a lack of separability is necessary for QKD. Now we prove that lack of separability is sufficient. 

\begin{lemma} \label{path non separable}
    Any two qubits of a graph state are non separable if and only if the two qubits are connected by a path in the graph representation of the state.
\end{lemma}
\begin{proof}
    In contrast to generic multi-qubit states, stabilizer states allow an efficient representation called the tableau or symplectic representation. For an $n$-qubit stabilizer state, there are $n$ independent stabilizer elements. Each stabilizer element can represented by a $2n$ length binary(bit) string. The first $n$ bits ($\overline{X}$)  denote which qubits Pauli $X$'s act on. The next $n$ bits ($\overline{Z}$) does the same for Pauli $Z$'s. For Pauli $Y$ on qubit $i$, we simply consider it as a product of $X$ and $Z$ and insert a $1$ at the $i$'th location for both $\overline{X}$ and $\overline{Z}$. We can construct this binary vector for each stabilizer generator ($S_i \rightarrow [\overline{X}_i \;|\; \overline{Z}_i]$) and stacking these row vectors on top of each other gives us a $2n\times n$ bit-matrix called the tableau representation of the stabilizer state.
    
    The tableau representation of a general stabilizer state takes the form,
    \begin{eqnarray*}
        T = [\hat{X}\;|\;\hat{Z}],
    \end{eqnarray*}
    where, $\hat{X}$ and $\hat{Z}$ are $n\times n$ matrices built from stacked $\overline{X}$ vectors and $\overline{Z}$ vectors respectively.
    We have already shown that the stabilizer group $(\mathcal{S})$ of a state that is separable across $A,B$ bipartition takes the form,
    \begin{eqnarray*}
        \mathcal{S} = \mathcal{S^{\prime}_A} \otimes \mathcal{S^{\prime\prime}_B}.
    \end{eqnarray*}
    For such a separable $\mathcal{S}$, if one constructs a tableau representation by stacking $2n$ length binary vectors corresponding to $\mathcal{S^{\prime}_A}$ on top of the vertically stacked binary vectors corresponding to $\mathcal{S^{\prime\prime}_B}$, it is clear that the tableau representation would be such that both $\hat{X}$ and $\hat{Z}$ are block diagonalizable with the two blocks corresponding to qubits in $A$ and $B$.
    \begin{eqnarray*}
        T = \left[\begin{array}{c c |c c}
            \hat{X}^{\prime}_A & \mathbf{0} & \hat{Z}^{\prime}_A & \mathbf{0} \\
             \mathbf{0} & \hat{X}^{\prime\prime}_B & \mathbf{0} & \hat{Z}^{\prime\prime}_B 
        \end{array}\right].
    \end{eqnarray*}
    It is also clear that block diagonalizability of the tableau representation implies that the stabilizer state is separable. For graph states, the tableau representation takes the form,
    \begin{eqnarray*}
        T = [I\;|\;G].
    \end{eqnarray*}
    $I$ is already in block diagonal form, and from elementary graph theory it can be shown that the adjacency matrix is block diagonalizable only if there exists disjoint sets of vertices with no path connecting them.
\end{proof}

Therefore, any graph state whose underlying graph is connected, cannot be separable.

\begin{lemma}\label{no substring using path}
  Secure QKD is possible between any two qubits  of a shared graph state if in the underlying graph, a path exists between the two qubits.   
\end{lemma}

\begin{proof}

Consider a graph state $\ket{G}$ with the stabilizer group $\mathcal{S}$ and stabilizer generators, $\{S_i\}$. With lemma \ref{path non separable},  we have already shown that if a path exists between qubits $i$ and $j$, it is non separable and vice versa. 
 
Without loss of generality, we re-label the two communicants to $1$ and $r$ and label the qubits on a shortest path between $1$ and $r$ as 

\begin{eqnarray*}
 p=\{1, 2, 3,\dots,r-2,r-1,r\}.   
\end{eqnarray*}

We use this path, $p$ to construct a stabilizer element $\sigma \in \mathcal{S}$ that satisfies fact \ref{substring condition}.

Now, consider a stabilizer element constructed by taking a product of all odd number labelled $S_i$'s along this shortest path, $p$,
\begin{equation}
    \sigma = S_1 \cdot S_{3} \cdot S_{5} \cdots S_{r-4}\cdot S_{r-2} \cdot S_{r}.
\end{equation}
We have assumed, without loss of generality, that $r$ is an odd number because, even if $r$ was an even number, the string of operators above would end in $S_{r-1}$, which also has support on $r$. Also, since $p$ is a shortest path and since each $S_i$ only has support on $i$ and its neighbours, $\sigma$ consists of Pauli $X$'s on the qubits $\{1,3,5,\dots\}$. All other qubits in the support of $\sigma$ are acted upon only by Pauli $Z$'s. A point to note is that due to overlapping Pauli $Z$'s
squaring to identity, $\sigma$ has no support on the qubits in $\{2,4,6,\dots\}$.

To satisfy  fact \ref{substring condition}, no substring of $\sigma$ with support on $i$ or $j$ should be a stabilizer element. We prove that $\sigma$ has no such substrings by assuming that a non-trivial stabilizer substring exists and arrive at a contradiction.

Let $\eta$ be a non-trivial substring of $\sigma$ with support on $i$ and $\eta \in \mathcal{S}$. The support of $\sigma$ is $\{1,3,5,\dots\}\cup A$, where $A = \mathcal{N}(1) \oplus \mathcal{N}(3) \oplus \mathcal{N}(5)\oplus \cdots$. Since $\eta$ is a substring of $\sigma$, $\text{support}(\eta) \subset \text{support}(\sigma)$. Since Pauli strings containing only $Z$ cannot be an element of $\mathcal{S}$, $\eta$ has non trivial support on some subset of the qubits in $\{1,3,5,\dots\}$. Let us denote this subset of $\{1,3,5,\dots\}$ with $B$. For any $\eta \in \mathcal{S}$, 

\begin{eqnarray*}
    \eta = \prod_{i} S_i^{r_i} \;\; ; \;\; r_i\in \{0,1\}.   
\end{eqnarray*}

But, since Pauli $X$'s of $\eta$ are all in $B$, $\eta = \prod_{i\in B} S_i$. Recall that we have indexed all qubits from $i$ to $j$ as $\{1,2,3,\dots,r\}$. Let us denote the qubit closest to $r$ in $B$ as $l$. $p$ being the shortest path implies that, 

\begin{eqnarray*}
 l + 1 \notin \bigcup_{i \in B\backslash l}\mathcal{N}(i).   
\end{eqnarray*}

Therefore $\eta$ has support on $l+1$ with a Pauli $Z$. Note that $l+1$ is an even number indexed qubit. But, $\sigma$ has no support on qubits in $\{2,4,6,\dots\}$, \textit{i.e.}, even numbered qubits. Therefore, $\eta$ cannot be a substring of $\sigma$ because it has support on a qubit on which $\sigma$ has no support.  

\end{proof}

Now we have everything needed to prove theorem \ref{main_result},
\begin{proof}
 For any stabilizer state there exists a local Clifford equivalent graph state \cite{10.5555/2011477.2011481,681315,1023317}. Therefore given a non separable stabilizer state, via local Clifford unitary transformations one can be arrive at a non separable graph state. Let this local Clifford transformation be $U$. This graph state has a shortest path connecting the two communicant qubits as per lemma \ref{path non separable}. As per lemma \ref{no substring using path}, this shortest path can be used to construct a stabilizer element ($\sigma$) that allows secure QKD. Clearly, $U \sigma U^{\dagger}$ is a secure stabilizer choice that satisfies fact \ref{substring condition}, that can be used for QKD between the two communicants.     
\end{proof}

\end{document}